\documentclass[a4paper,UKenglish,cleveref,autoref]{lipics-v2021}

\usepackage{amsmath,amssymb,amsthm}
\usepackage{hyperref}
\usepackage{todonotes,comment}

\usepackage{microtype}
\usepackage{amsthm}
\usepackage{microtype,comment}
\usepackage{enumerate}
\usepackage{boxedminipage}
\usepackage{tikz,tkz-graph}
\usepackage{hyperref}
\usepackage{amssymb,amsmath}
\usepackage{slashbox}
\usepackage[T1]{fontenc}
\usepackage{graphicx}
\input
\usetikzlibrary{arrows}
\usetikzlibrary{shapes}

\usetikzlibrary{shapes}
\usepackage{comment}

\newcommand{\NP}{{\sf NP}}

\newcounter{ctrclaim}[theorem]
\newcounter{ctrcase}[theorem]

\newcommand{\problemdef}[3]{
        \begin{center}
                \begin{boxedminipage}{.99\textwidth}
                        \textsc{{#1}}\\[2pt]
                        \begin{tabular}{ r p{0.8\textwidth}}
                                \textit{~~~~Instance:} & {#2}\\
                                \textit{Question:} & {#3}
                        \end{tabular}
                \end{boxedminipage}
        \end{center}
}

\def\romannum{\begingroup
  \def\theenumi{\textup{(\roman{enumi})}}%
  \def\p@enumi{}%
  \def\labelenumi{\theenumi}%
  \enumerate}

\title{Complexity Framework for Forbidden Subgraphs~III: When Problems are Tractable on Subcubic Graphs}

\titlerunning{Complexity Framework For Forbidden Subgraphs}

\author{Matthew Johnson}{Durham University, Durham, United Kingdom}{matthew.johnson2@durham.ac.uk}{}{}
\author{Barnaby Martin}{Durham University, Durham, United Kingdom}{barnaby.d.martin@durham.ac.uk}{}{}
\author{Sukanya Pandey}{Utrecht University, Utrecht, The Netherlands}{s.pandey1@uu.nl}{}{}
\author{Dani\"el Paulusma}{Durham University, Durham United Kingdom}{daniel.paulusma@durham.ac.uk}{}{}
\author{Siani Smith}{University of Bristol and Heilbronn Institute for Mathematical Research, Bristol, United Kingdom}{siani.smith@bristol.ac.uk}{}{}
\author{Erik Jan van Leeuwen}{Utrecht University, Utrecht, The Netherlands}{e.j.vanleeuwen@uu.nl}{}{}

\authorrunning{M. Johnson et al.}



\nolinenumbers
\hideLIPIcs

\begin{document}
\maketitle

\begin{abstract}
For any finite set $\mathcal{H} = \{H_1,\ldots,H_p\}$ of graphs, a graph is $\mathcal{H}$-subgraph-free if it does not contain any of $H_1,\ldots,H_p$ as a subgraph. In recent work, meta-classifications have been studied: these show that if graph problems satisfy certain prescribed conditions, their complexity is determined  on classes of $\mathcal{H}$-subgraph-free graphs.  We continue this work and focus on problems that have polynomial-time solutions on classes that have bounded treewidth or maximum degree at most~$3$ and examine their complexity on $H$-subgraph-free graph classes where $H$ is a connected graph. With this approach, we obtain comprehensive classifications for {\sc (Independent) Feedback Vertex Set}, {\sc Connected Vertex Cover}, {\sc Colouring} and {\sc Matching Cut}.  This resolves a number of open problems.  

We highlight that, to establish that {\sc Independent Feedback Vertex Set} belongs to this collection of problems, we first show that it can be solved in polynomial time on graphs of maximum degree $3$.  We demonstrate that, with the exception of the complete graph on four vertices, each graph in this class has a minimum size feedback vertex set that is also an independent set.
\end{abstract}

\section{Introduction}\label{s-intro}

A graph $G$ contains a graph $H$ as a {\it subgraph} if $H$ can be obtained from $G$ by  vertex deletions and edge deletions; else $G$ is said to be {\it $H$-subgraph-free}.  If $H$ can be obtained from $G$ using \emph{only} vertex deletions, then $H$ is an \emph{induced} subgraph of $G$, and if not then $G$ is \emph{$H$-free}.
There are few studies of complexity classifications of graph problems for $H$-subgraph-free graphs (compare the greater attention given to problems on $H$-free graphs). There are results for {\sc Independent Set}, {\sc Dominating Set} and {\sc Longest Path}~\cite{AK90}, Max-Cut~\cite{Ka12} and {\sc List Colouring}~\cite{GP14} 
In these papers, complete classifications are presented giving the complexity of the problems even for ${\cal H}$-subgraph-free graphs, where ${\cal H}$ is any finite set of graphs
(for a 
 set of graphs ${\cal H}$, a graph $G$ is  {\it ${\cal H}$-subgraph-free} if $G$ is $H$-subgraph-free for every $H\in {\cal H}$).
Such classifications seem difficult to obtain.   For example, for {\sc Colouring}, there is only a partial classification\cite{GPR15}. For this reason -- and also noting that the classifications for the problems above were all the same -- a systematic approach was developed in~\cite{JMOPPSV} with the introduction of a new framework which we will describe after introducing some terminology. 

For an integer $k\geq 1$, the {\it $k$-subdivision} of an edge $e=uv$ of a graph replaces $e$ by a path of length $k+1$ with endpoints $u$ and $v$ (and $k$ new vertices). 
The {\it $k$-subdivision} of a graph~$G$ is the graph obtained from $G$ after $k$-subdividing each edge. 
For a graph class ${\cal G}$ and an integer~$k$, let ${\cal G}^k$ consist of the $k$-subdivisions of the graphs in ${\cal G}$.
Let $\Pi$ be a graph problem.
We say that $\Pi$ is \NP-complete
{\it under edge subdivision of subcubic graphs} if there exists an integer~$k\geq 1$ such that the following holds for 
the class of
subcubic graphs ${\cal G}$:
if $\Pi$ is \NP-complete for ${\cal G}$, then $\Pi$ is \NP-complete for ${\cal G}^{kp}$ for every integer $p\geq 1$. 
A graph problem~$\Pi$ is a {\it C123-problem} (belongs to the framework) if it satisfies the three conditions:\\[-10pt]
\begin{enumerate}
\item [{\bf C1.}] $\Pi$ is polynomial-time solvable for every graph class of bounded treewidth;
\item [{\bf C2.}] $\Pi$ is \NP-complete
 for the class of subcubic graphs; and
\item [{\bf C3.}] $\Pi$ is \NP-complete under edge subdivision 
of subcubic graphs.\\[-10pt]
\end{enumerate}
\noindent
 As shown in~\cite{JMOPPSV}, C123-problems allow for full complexity classifications for ${\cal H}$-subgraph-free graphs (as long as ${\cal H}$ has finite size).
 A {\it subdivided} claw is a graph obtained from a claw ($4$-vertex star) after subdividing each of its edges zero or more times. The {\it disjoint union} of two vertex-disjoint graphs $G_1$ and $G_2$ has vertex set $V(G_1)\cup V(G_2)$ and edge set  $E(G_1)\cup E(G_2)$. The set ${\cal S}$ consists of the graphs that are disjoint unions of subdivided claws and paths. 
Now, let  $\Pi$ be a C123-problem. For a finite set ${\cal H}$, the problem $\Pi$ on ${\cal H}$-subgraph-free graphs is efficiently solvable if ${\cal H}$ contains a graph from ${\cal S}$ and computationally hard otherwise.~\cite{JMOPPSV}.

Examples of C123-problems include {\sc Independent Set}, {\sc Dominating Set}, {\sc List Colouring}, {\sc Odd Cycle Transversal}, {\sc Max Cut} and {\sc Steiner Tree}; see~\cite{JMOPPSV} for a comprehensive list. 
Thus we see the power of the framework to aid progress in deciding the complexity of problems on ${\cal H}$-subgraph-free graphs.   But
there are still many graph problems that are not C123.
In~\cite{MPPSV22}, results were obtained for problems that satisfy C1 and C2 but not C3.
Such problems are called C12-problems and include {\sc $k$-Induced Disjoint Paths}, {\sc $C_5$-Colouring}, {\sc Hamilton Cycle} and {\sc Star $3$-Colouring}~\cite{MPPSV22}.  
And in~\cite{BJMOPPSV23}, {\sc Steiner Forest} was investigated as a problem that satisfies C2 and C3 but not C1.  
We consider the research question:

\medskip
\noindent
{\it How do \emph{C13-problems} --- that is, problems that satisfy C1 and C3 but not C2 --- behave for $H$-subgraph-free graphs? Can we still classify their computational complexity?}

\medskip
\noindent
Let us immediately note some redundancy in the definition of C13-problems: if a problem does not satisfy C2, then C3 is implied.  Nevertheless we retain the terminology to preserve the link to the approach of~\cite{JMOPPSV}. 
To show a problem is a C13 problem there are two requirements: that the problem is efficiently solvable both on classes of bounded treewidth and on subcubic classes.  In fact, the tractable cases 
for C123 problems rely on that the problems satisfy C1. 

\begin{theorem}[\cite{JMOPPSV}]\label{t-dicho3}
Let $\Pi$ be a problem that satisfies C1.
For a finite set ${\cal H}$, the problem $\Pi$ on ${\cal H}$-subgraph-free graphs is efficiently solvable if ${\cal H}$ contains a graph from ${\cal S}$.
\end{theorem}

\noindent
As an important step towards a full dichotomy for C13 problems, we restrict ourselves to considering $H$-subgraph-free graphs where $H$ is connected. We focus on five well-known \NP-complete problems that we will see are not C123 but C13-problems: {\sc Feedback Vertex Set}, {\sc Independent Feedback Vertex Set}, {\sc Connected Vertex Cover} and {\sc Matching Cut}. 
We introduce these problems below.
With one exception, we can recognize that they are C13 problems using known results.   

For a graph $G=(V,E)$, a set $W \subseteq V$ is a \emph{feedback vertex set} of $G$ if every cycle in $G$ contains a vertex of $W$.  Moreover, $W$ is an \emph{independent feedback vertex set} if 
$W$ is an independent set.
We note that $G$ has a feedback vertex set of size $k$ if and only if the $2$-subdivision of $G$ has an independent feedback vertex set of size~$k$.
A graph $G$ might contain no independent feedback vertex set: consider, for example, a complete graph on four or more vertices. 
The {\sc (Independent) Feedback Vertex Set} problem is to decide if a graph $G$ has an (independent) feedback vertex set of size at most~$k$ for some given integer $k$.

A set $W \subseteq V$ is a {\it connected vertex cover} of $G$ if every edge of $E$ is incident with a vertex of~$W$, and moreover $W$ induce a connected subgraph. The {\sc Connected Vertex Cover} problem is to decide if a graph $G$ has a connected vertex cover of size at most~$k$ for a given integer $k$.
A \emph{$k$-colouring} of $G$ is a function $c: V \rightarrow \{1, \ldots ,k\}$ such that for each edge $uv \in E$, $c(u) \neq c(v)$. The {\sc Colouring} problem is to decide if a graph $G$ has a $k$-colouring for some given integer~$k$. A \emph{matching cut} of a connected graph is a matching (set of pairwise non-adjacent edges) that is also an edge cut, i.e., its removal creates a disconnected graph. 
The {\sc Matching Cut} problem is to decide if a connected graph has a matching cut.

\subsection{Our Results}

Whereas {\sc Feedback Vertex Set} does have a polynomial-time algorithm on subcubic graphs~\cite{UKG88} and thus does not satisfy C2, a polynomial-time algorithm for {\sc Independent Feedback Vertex Set} on subcubic graphs was not previously known.
In Section~\ref{s-i}, we prove the following result addressing this gap in the literature.

\begin{theorem} \label{t-subcubic}
A minimum size independent feedback vertex set of every connected subcubic graph $G\neq K_4$ is also a minimum size feedback vertex set of $G$. Moreover, it is possible to find a minimum independent feedback vertex set of $G$ in polynomial time.
\end{theorem}

\noindent
Hence, both {\sc Feedback Vertex Set} and {\sc Independent Feedback Vertex Set} are C13.
The other problems are also C13. Namely, {\sc Connected Vertex Cover} satisfies C1~\cite{ALS91} and is polynomial-time solvable on subcubic graphs~\cite{UKG88} so does not satisfy C2, while 
{\sc Colouring} also satisfies C1~\cite{ALS91} but not C2 due to Brooks' Theorem~\cite{Brooks41}. Finally, {\sc Matching Cut} satisfies C1~\cite{Bo09} but not C2, due to a polynomial-time algorithm for subcubic graphs~\cite{Ch84}.

The \emph{star} $K_{1,s}$ is the graph that contains a vertex of degree $s$ whose neighbours each have degree~$1$.  A \emph{subdivided star} is obtained from a star by subdividing one or more of its edges.

\begin{definition}\label{def-claw}
An $S_{w,x,y,z}$ is a graph formed by subdividing each edge of a $K_{1,4}$, $w-1, x-1, y-1,$ and $z-1$ times. Each of the subdivided edges is called a \emph{tentacle}.   The vertex of degree~$4$ is the \emph{centre}.  
\end{definition}

\noindent
In Section~\ref{s-treedepth}, we investigate the structure of $H$-subgraph-free graphs when $H$ is a subdivided star and use this in Section~\ref{s-algos} to 
show a general approach to C13 problems that requires some additional extra properties (that they can be solved componentwise after, possibly, the removal of bridges).  This is sufficient to obtain the following result.

\begin{theorem} \label{cor:algo:s11qr}
Let $q$ and $r$ be positive integers.
The following problems can be solved in polynomial time on $S_{1,1,q,r}$-subgraph-free graphs: {\sc Feedback Vertex Set}, {\sc Independent Feedback Vertex Set}, {\sc Connected Vertex Cover}, {\sc Colouring} and {\sc Matching Cut}.
\end{theorem}

\noindent
In Section~\ref{s-hardness}, we obtain a hardness result.

\begin{theorem}\label{thm:hardness}
	{\sc Feedback Vertex Set} and {\sc Independent Feedback Vertex Set} are \NP-complete on the class of  $S_{2,2,2,2}$-subgraph free graphs that have maximum degree~$4$.
\end{theorem}

\subsection{State-of-the-Art Summaries}

We now state
complexity classifications for each of the problems. These results, proved in Section~\ref{s-proofs}, combine the results above with a number of other results from \cite{Ch84,FLPR23,GJ79,GJS76,GPR15,MP96,Munaro17,Poljak1974,Sp83}.  None of these papers presented general results for C13 problems. However, we note, for example, that hardness when $H$ contains a cycle follows from past results on classes of bounded girth which were proved separately for each problem, but often using a similar technique.  There are other results that just apply to one or two of the problems.

\begin{theorem}\label{t-dichofvs}
 Let $H$ be a connected graph.
On $H$-subgraph-free graphs,  {\sc Feedback Vertex Set} and {\sc Independent Feedback Vertex Set} are solvable in polynomial time if $H \in {\cal S }  \cup \{S_{1,1,q,r} \mid q \geq r \geq 1\}$.  They are \NP-complete if $H$ contains a cycle or more than one vertex of degree at least $3$ or $H \in \{K_{1,5}, S_{2,2,2,2}\}$.
\end{theorem}


\begin{theorem}\label{t-dichocvc}
 Let $H$ be a connected graph.
On $H$-subgraph-free graphs,  {\sc Connected Vertex Cover} is solvable in polynomial time if $H \in {\cal S }  \cup \{S_{1,1,q,r} \mid q \geq r \geq 1\}$.  It is \NP-complete if $H$ contains a cycle or $H=K_{1,5}$.
\end{theorem}

\noindent
The following result refers to trees defined in Figure~\ref{fig:T}.

\begin{theorem}\label{t-dichocol}
 Let $H$ be a connected graph.
On $H$-subgraph-free graphs,  {\sc Colouring} is solvable in polynomial time if $H \in {\cal S } \cup \{S_{1,1,q,r} \mid q \geq r \geq 1\}$ or if $H$ is a forest  with maximum degree~$4$ and at most seven vertices.  It is \NP-complete if $H$ contains a cycle, or $H \in \{K_{1,5}, S_{2,2,2,2}\}$, or if  $H$ contains a subdivision of the tree $T_1$  as a subgraph, or~$H$ contains as a subgraph the tree obtained from  $T_2$ after subdividing the edge $st$ $p$ times, $0\leq p\leq 9$, or $H$ contains one of the trees $S_{2,2,2,2},T_4,T_5,T_6$ as a subgraph.
\end{theorem}

\begin{figure}[ht]
\centering
		\includegraphics[width=0.75\textwidth]{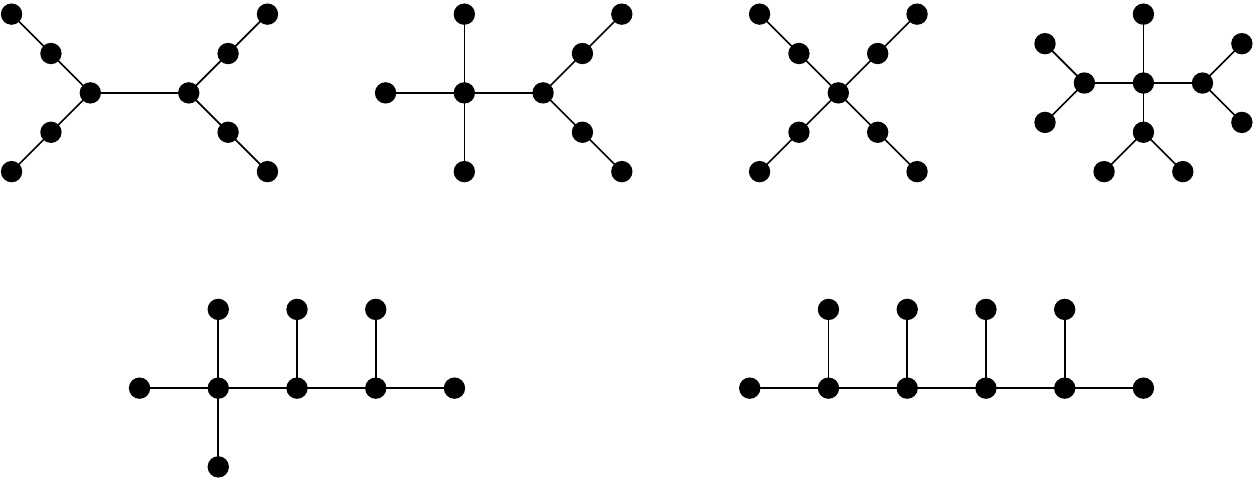}
\caption{Illustration of the trees $T_1,\ldots, T_6$ reproduced from~\cite{GPR15}; note that $T_3=S_{2,2,2,2}$.
\label{fig:T}}
\end{figure}

\begin{theorem}\label{t-dichomc}
 Let $H$ be a connected graph.
On $H$-subgraph-free graphs,  {\sc Matching Cut} is solvable in polynomial time if $H \in {\cal S }  \cup \{S_{1,1,q,r} \mid q \geq r \geq 1\}$.  It is \NP-complete if $H$ contains a cycle or $H=K_{1,5}$.
\end{theorem}

\section{Independent Feedback Vertex Sets of Subcubic Graphs}\label{s-i}

In~\cite{UKG88}, Ueno, Kajitani and Gotoh gave a polynomial-time algorithm for {\sc Feedback Vertex Set} on subcubic graphs.  In this section, we prove Theorem~\ref{t-subcubic} by showing that {\sc Independent Feedback Vertex Set} is also polynomially-time solvable on subcubic graphs
 by demonstrating that the problems are alike as, for any subcubic graph, one can find a minimum size feedback vertex set that is also an independent set (with a single exceptional case).  As the problems can be solved componentwise, we consider only connected graphs. 
 
In fact, we are going to prove a result that is an expansion of Theorem~\ref{t-subcubic} that will come in handy later.   We need some definitions.  A \emph{cactus} is a graph in which no two cycles have an edge in common.  A cactus is \emph{nice} if no two cycles have a vertex in common (every subcubic cactus is nice since if two cycles share a vertex but not an edge, we can find a vertex of degree~$4$).  A cactus is \emph{very nice} if every vertex belongs to exactly one cycle.
 
\begin{theorem} \label{t-subcubic2}
Let $G$ be a connected subcubic graph.  Then a minimum size independent feedback vertex set of $G$ is also a minimum size feedback vertex set of $G$ if and only if $G \neq K_4$.  Moreover, if $G \neq K_4$ there is a minimum size independent feedback vertex set of $G$ that contains only vertices of degree~$3$ if and only if $G$ is not a very nice cactus.   There is a polynomial-time algorithm to find a minimum size independent feedback vertex set and if $G$ is not a very nice cactus it finds a set that contains only vertices of degree~$3$.
\end{theorem}

\begin{proof}
It will be seen that the proof implies a polynomial-time algorithm for finding an independent feedback vertex set of the size no greater than a given feedback vertex set.

A feedback vertex set of $K_4$ must contain at least two vertices and so~$K_4$ has no independent feedback vertex set.  In a very nice cactus, the minimum size of a feedback vertex set is equal to the number of cycles and one can easily find such a set that is independent if one permits the inclusion of degree 2 vertices.  (For example, pick an arbitrary vertex $v$ and form an independent feedback vertex set by taking the vertex in each cycle that is farthest from $v$.)  If there are $k$ (disjoint) cycles, then, considering the tree-like structure of a very nice cactus, there are $2(k-1)$ vertices of degree~$3$ that can be considered as $k-1$ adjacent pairs.  Thus no set of $k$ vertices of degree~$3$ is independent.

So suppose that $G \neq K_4$ is not a very nice cactus.  Of course, we may as well also assume that $G$ is not a tree.  Let $F$ be a feedback vertex set of~$G$.  To prove the theorem, we must show that we can find an independent feedback vertex set of~$G$ that is no larger than $F$.  We can assume that $F$ contains only vertices of degree~$3$ since any vertex of degree 2 can be replaced by a nearest vertex of degree~$3$.  As $G$ is neither a tree nor a cycle (a cycle is a very nice cactus), we know that $G$ has vertices of degree~$3$.

Let $J = \emptyset$.   Our approach is to add vertices to $J$ until it forms an independent feedback vertex set.
We make some trivial but useful statements:
\begin{enumerate}
\item $F$ is a feedback vertex set containing only vertices of degree~$3$,
\item $J \subseteq F$, and
\item $J$ is a nonseparating independent set of $G$; that is, no pair of vertices of $J$ are joined by an edge and $G-J$ is connected.
\end{enumerate}
We will repeatedly modify $F$ and $J$ in such a way that these three statements remain true and the size of $F$ does not increase and it remains a feedback vertex set. 
We can make the following changes without contradicting the three statements. 
\begin{itemize}

\item  We can add a vertex $x \in F \setminus J$ to $J$ if $x$ has no neighbour in $J$ and is not a cutvertex in~$G-J$.
\item If $x \in F \setminus J$, we can redefine $F$ as $F \setminus \{x\} \cup \{y\} $ if $y$ is a vertex 
that belongs to every cycle of $G - (F \setminus \{x\})$ and has degree~$3$ (that is, $y$ belongs to every cycle of $G$ that contains $x$ but no other vertex of $F$).
\end{itemize}

 Our initial aim is to make changes so that $G-J$ is a graph where no two cycles have a vertex in common; that is, it is a nice cactus.


\begin{claim}  \label{claim:cactus}
We can modify $F$ and $J$ until $G-J$ is a nice cactus.
\end{claim}
\begin{claimproof}
Assume $G-J$ contains two cycles with a common vertex, and, therefore, as $G$ is subcubic, a common edge, else we are done.  
Consider a subgraph $K$ induced by two cycles of $G-J$ that have a common edge (so $K$ is 2-connected and has no cutvertex). Of course,~$F$ must contain at least one vertex of $K$; let $r$ be such a vertex.
 
If $r$ has degree~$3$ in $K$, then we can add it to $J$ since it has three neighbours in $G-J$ (so none in $J$) and is not a cutvertex in $G-J$ since $K - \{r\}$ is connected.  


Otherwise $r$ has degree 2 in $K$.
Traversing edges of $K$ away from $r$ in either direction, let~$p$ and $q$ be the first vertices of degree~$3$ in $K$ that are reached (and $p \neq q$ by the definition of~$K$). 
   Let $r'$ be the first vertex of degree~$3$ in $G$ reached from $r$ on the path in $K$ towards $p$.  
   
   If $r$ has a neighbour $j \in J$, then we can redefine~$F$ as $F \setminus \{r\} \cup \{r'\}$ since every cycle in~$G$ containing $r$ also contains either~$j$ or~$r'$.  Suppose instead that $r$ has no neighbour in $J$.  Let~$r''$ be the neighbour of $r$ in $G-J$ but not $K$.  If $r$ is not a cutvertex in $G-J$, then we can add $r$ to $J$.  If $r$ is a cutvertex in $G-J$, then no cycle of $G-J$ includes the edge $rr''$.  Thus, again, we can redefine $F$ as $F \setminus \{r\} \cup \{r'\}$.  
   
   So we either add a vertex to $J$ or modify $F$ by replacing a vertex with another that is closer in $K$ to $p$.  By repetition, we either add a vertex to $J$ or modify $F$ to include~$p$ in which case, as noted above, we can add $p$ to $J$.  Therefore, if $G-J$ contains two cycles with a common edge, we can increase the size of $J$ and so, ultimately, we can assume that $G-J$ contains no such pair of cycles and is a nice cactus.  This completes the proof of Claim~\ref{claim:cactus}.
\end{claimproof}

Let $H= G-J$. . 
By Claim~\ref{claim:cactus},
the cycles of $H$ are vertex disjoint and the graph has a treelike structure: if one replaces each cycle by a single vertex, then a tree is obtained.  As $F$ must contain at least one vertex of each cycle of $H$, if we add to $J$ one vertex chosen from each cycle of~$H$ (in any way), it will be no larger than $F$.   If we can do this in such a way that $J$ is an independent set and each vertex has degree~$3$, then the proof will be complete.
Thus we must describe how to choose a degree~$3$ vertex from each cycle of $H$ such that the union of these vertices and $J$ is an independent set, possibly after some further minor modifications.
The reasoning about these modifications will require that $H$ is connected so the requirement above that $J$ be nonseparating was needed.

If $H$ contains no cycles, then $J$ is already an independent feedback vertex set and there is nothing to prove.
Otherwise, let $C$ be a cycle of $H$.  Let $S(C)$ be the set of vertices that contains, for each cycle $C'$ of $H$ other than $C$, the vertex of $C'$ that is nearest to $C$ in~$H$.  See Figure~\ref{fig-subcubiccactus}. Each vertex $v$ of $S(C)$ has degree~3 in $H$ since it has two neighbours in a cycle $C'$ and a neighbour not in $C'$ on the  
path from $v$ to $C$.  Thus no vertex of $S(C)$ has a neighbour in $J$.  Moreover, clearly $S(C)$ is an independent set.  Thus $J \cup S(C)$ is an independent set that covers every cycle of $G$ except $C$.   
For a vertex $v$ in $C$, let $F(v)=J \cup S(C) \cup \{v\}$.
 If we can find a cycle $C$ that contains a vertex $v$ of degree~$3$ not adjacent to~$J$ or to another cycle in $H$, then $F(v)$ is an independent feedback vertex set and we are done.

\tikzstyle{vertex}=[circle,draw=black, minimum size=6pt, inner sep=0pt]
\tikzstyle{edge} =[draw,thick,-,black,>=triangle 90]

\begin{figure} 
\begin{center}
\begin{tikzpicture}[scale=1]

     \node at  (8,4) {$C$};

\begin{scope}[xshift=8cm, yshift=4cm] 
\foreach \pos/ \name in {{(0,1)/c1}, {(0.951,0.309)/c2}, {(0.588,-0.809)/c3},
                            {(-0.588,-0.809)/c4}, {(-0.951,0.309)/c5}}
        \node[vertex, fill=black]  (\name) at \pos {};
     \foreach \source/ \dest  in {c1/c2, c2/c3, c3/c4, c4/c5, c5/c1}
       \path[edge] (\source) --  (\dest);
\end{scope}

\begin{scope}[xshift=7.5cm, yshift=6cm] 
\foreach \pos/ \name in {{(0,1)/2}, {(1,1)/3},
                            {(1,0)/4}}
        \node[vertex, fill=black]  (\name) at \pos {};
        \node[vertex, fill=white]  (a1) at (0,0) {};

     \foreach \source/ \dest  in {a1/2, 2/3, 3/4, 4/a1}
       \path[edge] (\source) --  (\dest);
\end{scope}

\begin{scope}[xshift=11cm, yshift=5.5cm] 
\foreach \pos/ \name in {{(0,1)/2}, {(1,1)/3},
                            {(1,0)/4}}
        \node[vertex, fill=black]  (\name) at \pos {};
        \node[vertex, fill=white]  (b1) at (0,0) {};

     \foreach \source/ \dest  in {b1/2, 2/3, 3/4, 4/b1}
       \path[edge] (\source) --  (\dest);
\end{scope}

\begin{scope}[xshift=5cm, yshift=5.5cm] 
\foreach \pos/ \name in {{(0,0)/e4}, {(0.5,0.87)/3}}
        \node[vertex, fill=black]  (\name) at \pos {};
                \node[vertex, fill=white]  (a2) at (1,0) {};
     \foreach \source/ \dest  in {e4/a2, a2/3, 3/e4}
       \path[edge] (\source) --  (\dest);
\end{scope}

\begin{scope}[xshift=3cm, yshift=5.5cm] 
\foreach \pos/ \name in {{(0,0)/1}, {(0.5,0.87)/3}}
        \node[vertex, fill=black]  (\name) at \pos {};
                \node[vertex, fill=white]  (b2) at (1,0) {};
     \foreach \source/ \dest  in {1/b2, b2/3, 3/1}
       \path[edge] (\source) --  (\dest);
\end{scope}

     \node at  (6,4.3) {$w$};
     \node at  (7.4,2.9) {$v$};

        \node[vertex, fill=black]  (w1) at (6,4) {};

        \node[vertex, fill=black]  (w) at (6,3) {};
        \node[vertex, fill=black]  (x) at (5,3) {};

        \node[vertex, fill=black]  (w2) at (5,4.5) {};
        \node[vertex, fill=black]  (w3) at (13,5.3) {};
\node[vertex, fill=black]  (w4) at (13,6.3) {};

\begin{scope}[xshift=3cm, yshift=3cm] 
\foreach \pos/ \name in {{(0,0)/1}, {(0.5,-0.87)/3}}
        \node[vertex, fill=black]  (\name) at \pos {};
                \node[vertex, fill=white]  (f2) at (1,0) {};
     \foreach \source/ \dest  in {1/f2, f2/3, 3/1}
       \path[edge] (\source) --  (\dest);
\end{scope}

\begin{scope}[xshift=10cm, yshift=2.5cm] 
\foreach \pos/ \name in {{(0.5,0.87)/g2}, {(1,0)/g3}}
        \node[vertex, fill=black]  (\name) at \pos {};
                \node[vertex, fill=white]  (g1) at (0,0) {};
     \foreach \source/ \dest  in {g1/g2, g2/g3, g3/g1}
       \path[edge] (\source) --  (\dest);
\end{scope}

\begin{scope}[xshift=12cm, yshift=3.5cm] 
\foreach \pos/ \name in {{(0,1)/2}, {(1,1)/k3},
                            {(1,0)/4}}
        \node[vertex, fill=black]  (\name) at \pos {};
        \node[vertex, fill=white]  (h1) at (0,0) {};

     \foreach \source/ \dest  in {h1/2, 2/k3, k3/4, 4/h1}
       \path[edge] (\source) --  (\dest);
\end{scope}

\begin{scope}[xshift=12cm, yshift=1.5cm] 
\foreach \pos/ \name in {{(0,0)/2}, {(1,1)/3},
                            {(1,0)/4}}
        \node[vertex, fill=black]  (\name) at \pos {};
        \node[vertex, fill=white]  (i1) at (0,1) {};

     \foreach \source/ \dest  in {i1/2, 4/2, 3/4, 3/i1}
       \path[edge] (\source) --  (\dest);
\end{scope}

    \foreach \source/ \dest  in {c1/a1, c2/b1, a2/c5, e4/b2, x/w, w/c4, f2/x, c3/g1, h1/g2, g3/i1, w/w1, w2/e4, w3/w4, k3/w3}
      \path[edge] (\source) --  (\dest);

\end{tikzpicture}
\end{center}
\caption{A nice subcubic cactus.  The central 5-cycle is denoted $C$ and the white vertices form the set $S(C)$.  Note that $w$ does not belong to any cycle and~$v$ is the nearest vertex  to $w$ in a cycle.  Thus $S(C) \cup \{v\}$ is an independent feedback vertex set for the graph.}
\label{fig-subcubiccactus}
\end{figure}
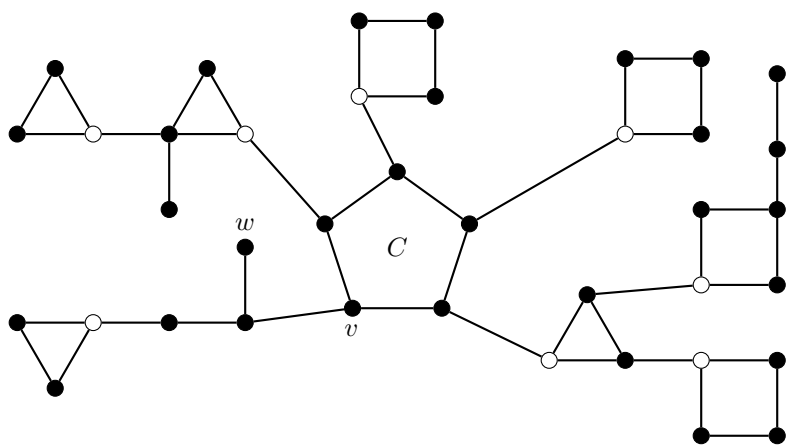

Suppose instead that no such cycle can be found.  Notice that this implies that every vertex of $H$ belongs to a cycle.  (If there was a vertex $w$ not in a cycle, then let $v$ be a nearest vertex to $w$ in a cycle and then $F(v)$ is an independent feedback vertex set of degree~$3$ vertices; again, see Figure~\ref{fig-subcubiccactus}.)  So $H$ is a very nice cactus and, by assumption, $J \neq \emptyset$.

Let $j$ be a vertex in $J$ with neighbours $v_1$, $v_2$ and $v_3$ in $H$.  Suppose that these three vertices are in the same cycle $C$ of $H$.  If $C$ is a 3-cycle, then $\{j,v_1,v_2,v_3\}$ induces $K_4$, a contradiction.  So we can assume that $v_1$ and $v_2$ are not adjacent.  Then $J_1 = J \setminus \{j\} \cup \{v_1, v_2\} \cup S(C)$ is an independent feedback vertex set of degree~$3$ vertices and $|J_1|=|F|$.  Quick check: all cycles are covered by $J_1$ since $v_1$ and $S(C)$ cover the cycles of $H$ and every cycle containing $j$ includes at least one of $v_1$ and $v_2$; $J_1$ is independent as $v_1$ and~$v_2$ have degree 2 in $H$ so no other neighbour in $J$ and are not adjacent to vertices, such as those of $S(C)$, that do not belong to $C$, and the vertices of $S(C)$ have degree~$3$ in $H$ so no neighbours in $J$.

Suppose instead that $v_1$, $v_2$ and $v_3$ do not all belong to the same cycle.  Let $C$ be the cycle that contains $v_1$ and suppose that $v_2$ and $v_3$ do not  belong to the same cycle as each other (one might belong to $C$).
Then $J_2 = J \setminus \{j\} \cup \{v_1\} \cup S(C)$ is an independent feedback vertex set of degree~$3$ vertices and $|J_1|=|F|-1$.  Quick check: all cycles are covered by $J_2$ since $v_1$ and $S(C)$ cover the cycles of~$H$ and every cycle containing $j$ includes either~$v_1$ or both $v_2$ and $v_3$ and all the paths from $v_2$ to $v_3$ (that do not include $j$) go through either a vertex of $J$ or a vertex of $S(C)$ as they are in different cycles in $H$; $J_2$ is independent as $v_1$ has degree 2 in~$H$ so, as before, no other neighbour in $J$ or $S(C)$, and the vertices of $S(C)$ have degree~$3$ in~$H$ so no neighbours in $J$.
 \end{proof}

\section{Graphs Excluding Subdivided Stars as a Subgraph: Structure}  \label{s-treedepth}

Recall that the treedepth of a graph $G$ is the minimum height of a forest $F$ such that for every pair of vertices in $G$ one is the ancestor of the other in $F$. 
It is well-known that the treewidth of a graph is at most its treedepth.  
In this section, we aim to show that $H$-subgraph-free graphs, for certain $H$, have bounded treedepth.  Then we know that problems that are tractable on classes of bounded treewidth are also tractable on these classes.
Before presenting our results, we need the following result from~\cite{NOM12}.

\begin{theorem}[\cite{NOM12}] \label{thm:longpath}
Let $G$ be a graph of treedepth at least $d$. Then $G$ has a subgraph isomorphic to a path of length at least $d$.
\end{theorem}

Our next two theorems consider graphs $S_{w,x,y,z}$.  By Definition~\ref{def-claw}, this graph is four paths sharing an endvertex. 
In a small abuse of terminology, we will use \emph{leaf} to mean only a vertex of degree~1 that is adjacent to the centre.  

\begin{theorem}\label{thm:s111r}
Let $r$ be a positive integer.  Then the subclass of connected $S_{1,1,1,r}$-subgraph-free graphs that are not subcubic has bounded treedepth.
\end{theorem}

\begin{proof}
Let $G$ be a connected $S_{1,1,1,r}$-subgraph-free graph that is not subcubic so contains a vertex $v_0$ with neighbours $v_1, v_2, v_3, v_4$.  We will show that $G$ has treedepth at most $2r+2$.
Suppose instead that the treedepth of $G$ is at least $2r+3$.   The graph $G \setminus \{v_0, v_1, v_2, v_3, v_4\}$ must have treedepth at least $2r-2$ (since adding a vertex to a graph cannot increase the treedepth by more than one),
and therefore, by Theorem~\ref{thm:longpath}, it must contain a path $P$ of length at least $2r-2$.
Let $Q$ be a shortest path in $G$ between $P$ and $v_0$ (which must exist as $G$ is connected). Let $z$ be the vertex where $P$ and $Q$ meet.  Let $P'$ be the longest subpath of $P$ of which~$z$ is an endvertex.  As $P'$ is at least half the length of $P$, and $Q$ contains at least one edge, the path $P' \cup Q$ contains at least $r$ edges.
Thus there exists in $G$ a subgraph isomorphic to $S_{1,1,1,r}$; the centre is $v_0$, $P' \cup Q$ is the tentacle of length $r$, and three of $v_1, v_2, v_3, v_4$ are the three leaves (since at most one of these four vertices can belong to $Q$ and none belong to~$P'$).  This contradiction completes the proof.
\end{proof}

The assumption that the graphs are connected is needed: the class of all graphs that are each a disjoint union of a path and a $K_{1,4}$ is not subcubic but has unbounded treedepth.

Consider now the class of all connected graphs that are each the union of a  path and a $K_{1,4}$, one of whose leaves is identified with the endvertex of the path.  This is a class of graphs that are connected, not subcubic and $S_{1,1,q,r}$-subgraph-free and again has unbounded treedepth.  Thus, in the following analogue of Theorem~\ref{thm:s111r}, we need an additional property.
A bridge is \emph{proper} if neither incident vertex has degree~$1$.
A graph is \emph{quasi-bridgeless} if it contains no proper bridge.

\begin{theorem}\label{thm:s11qr}
Let $q$ and $r$ be positive integers.  Then the subclass of  connected $S_{1,1,q,r}$-subgraph-free graphs that are not subcubic and are quasi-bridgeless has bounded treedepth.
\end{theorem}
\begin{proof}
Let $G$ be a connected quasi-bridgeless $S_{1,1,q,r}$-subgraph-free graph that is not subcubic so contains a vertex $v_0$ with neighbours $v_1, v_2, v_3, v_4$.  We will show that $G$ has treedepth at most $2(q+r+3)^2+6$.
Suppose instead that the treedepth of $G$ is at least $2(q+r+3)^2+7$.   The graph $J=G \setminus \{v_0, v_1, v_2, v_3, v_4\}$ must have treedepth at least $2(q+r+3)^2+2$  
and therefore, by Theorem~\ref{thm:longpath}, it must contain a path $P$ of length at least $2(q+r+3)^2+2$.
Let $z$ be the middle vertex of $P$. We prove the following claim.

\begin{claim} \label{c:cycleS}
If there is a cycle $C$ in $G$ that contains $z$ and also a vertex $v \not= z$ that has two neighbours $a$ and $b$ not on $C$, then $G$ contains a subgraph isomorphic to $S_{1,1,q,r}$.
\end{claim}
\begin{claimproof} 
A \emph{big adorned cycle} is a graph that contains a cycle with at least $q+r+1$ edges and two further vertices each joined by an edge to the same vertex on the cycle; the latter vertex is called the centre.   If we find a big adorned cycle in $G$ we are done as it contains a subgraph isomorphic to $S_{1,1,q,r}$ (the centre is the same and it is obtained by deleting one or more edges of the cycle).
Let $C^+$ be the union of $C$ and the vertices $a$ and $b$ and the edges $va$ and $vb$.  If $|C| \geq q+r+1$, then $C^+$ is a big adorned cycle.

So suppose that $|C| \leq q+r$.  Consider the intersections of $P$ with $V(C^+)$. A maximal subpath of~$P$ whose internal vertices are not in $V(C^+$) is called an \emph{interval} of $P$. Note that~$P$ has at most $|C^+|+1 \leq q+r+3$ intervals. If all intervals of $P$ have length at most $q+r-1$, then $P$ itself has length at most $(q+r+3) (q+r-1) < (q+r+3)^2$, a contradiction. Hence, at least one of the intervals  has length at least $q+r$; we call such an interval \emph{long}.  See Figure~\ref{fig:intervals} for an illustration.

	\begin{figure}[htbp]
		\centering
		\includegraphics[width=0.75\textwidth]{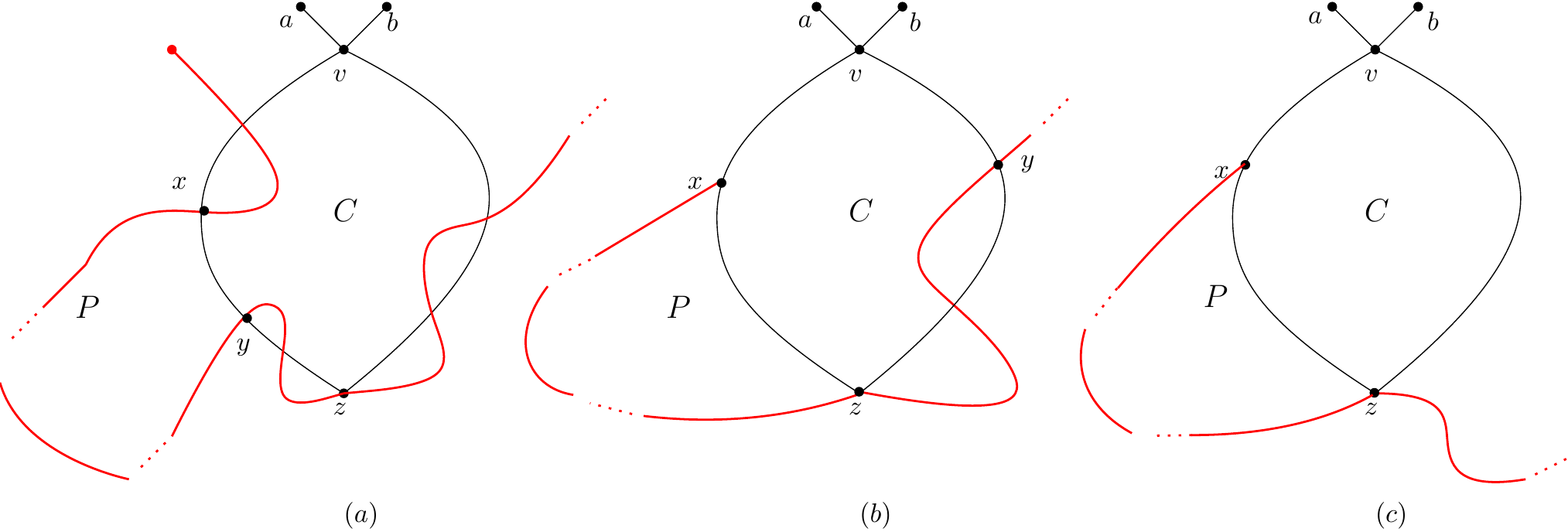}
		\caption{The cycle $C$ and path $P$ from the proof of Claim~\ref{c:cycleS} illustrating the three cases (a) there is a long interval with both endvertices in $P$, (b) there are two vertex-disjoint long intervals, and (c) there are two long intervals that meet in a single vertex.}   \label{fig:intervals}
	\end{figure}

Suppose that there is a long interval $L$ of which both endvertices $x$ and $y$ are in $V(C^+)$. Then there are shortest (possibly trivial) paths $S$ and $T$ on $C^+$ from $v$ to $x$ and $y$ respectively that are vertex disjoint except for $v$.   As $x$ and $y$ are distinct, the union of $L$, $S$ and $T$ is a cycle on at least $q+r+1$ edges.  As $v$ has four neighbours in $C^+$, two of them do not belong to this cycle and considering these two neighbours (and the incident edges that join them to $v$) with the cycle, we have a big adorned cycle centred at $v$.

Hence, there is no long interval with both endvertices in $V(C^+)$ and we can assume any long interval has just one endvertex in $V(C^+)$.
 Suppose that there are two long intervals $L_1$ and $L_2$ whose endvertices in $C^+$ are $x$ and $y$ respectively.  If $x=y$, then $L_1$ and $L_2$ are the only intervals and their union is $P$.  This implies that $P$ only intersects $C^+$ in $x$ and so we must have $x=z$.  Then there exists in $G$ a subgraph isomorphic to an $S_{1,1,q,r}$ with $z$ as its centre, the neighbours of $z$ on $C$ as the leaves and subpaths of $L_1$ and $L_2$ as the tentacles.  If $x \neq y$, then there are shortest paths $S$, $T$ on $C^+$ from $v$ to $x$ and $y$ respectively that are vertex disjoint except for $v$. Then there exists in $G$ a subgraph isomorphic to an $S_{1,1,q,r}$ with $v$ as its centre, the paths $S$ and $T$, possibly extended by subpaths of $L_1$ and $L_2$,  as the tentacles and two neighbours of $v$ in $C^+$ that do not belong to $S$ or $T$ as the leaves. 

Hence, there is only one long interval $L$.  As the other intervals are short, they have total length at most $|C^+| \cdot (q+r) < (q+r+3)^2$. Hence, $L$ has length at least $(q+r+3)^2+2$.  As $L$ contains more than half the vertices of $P$, the middle vertex of $P$ is an internal vertex of $L$ and so does not belong to $C^+$.  This contradicts that $z$ is the middle vertex of $P$ and completes the proof of the claim.
\end{claimproof}

\noindent
We now apply the claim.
Note that $v_0$ and $z$ are distinct as $z$ belongs to $J$ but $v_0$ does not.
 Since $G$ is quasi-bridgeless and neither $v_0$ nor $z$ has degree 1, it follows from Menger's Theorem~\cite{Menger1927} that there exist two edge-disjoint paths $S$, $T$ from $v_0$ to $z$. If $S$ and $T$ are internally vertex-disjoint paths, then their union forms a cycle that contains $z$.  We can assume that each of $S$ and $T$ contain only one neighbour of $v_0$ else we can find shortcuts and redefine them.  Hence, $v_0$ has two neighbours not in the cycle and we can apply Claim~\ref{c:cycleS}. 
If $S$ and $T$ are not internally vertex-disjoint, let $v'$ be a vertex of $(V(S) \cap V(T)) \setminus \{z\}$ that is furthest from $v$ on $T$. Consider the subpath $T'$ of $T$ from $v'$ to $z$ and the subpath $S'$ of $S$ from $v'$ to $z$. Since $T'$ does not intersect $S$ by definition, $S'$ and $T'$ are internally vertex disjoint. Hence, their union forms a cycle that contains $z$.  Moreover, $v'$ has degree at least four, of which two neighbours are not on $S'$ or $T'$. Hence, we can apply Claim~\ref{c:cycleS}.
\end{proof}

\section{Graphs Excluding Subdivided Stars as a Subgraph: Algorithms}  \label{s-algos}
We present several applications of the structural results of the previous section.

We note that {\sc Feedback Vertex Set}, {\sc Independent Feedback Vertex Set} and {\sc Colouring} can be solved componentwise.  In a sense, so can {\sc Connected Vertex Cover} and {\sc Matching Cut} since disconnected graphs are no instances (except possibly for {\sc Connected Vertex Cover} instances with edgeless components but these can be ignored).


\begin{theorem} \label{thm:algo:s111r}
Let $r$ be a positive integer.
A problem $\Pi$ can be solved in polynomial time on $S_{1,1,1,r}$-subgraph-free graphs if the following hold:
\begin{enumerate}[i)]
\item $\Pi$ can be solved in polynomial time on subcubic graphs,
\item $\Pi$ can be solved in polynomial time on graphs of bounded treedepth, and
\item $\Pi$ can be solved componentwise on disconnected graphs.
\end{enumerate}
\end{theorem}
\begin{proof}
Let $C$ be a connected component of a $S_{1,1,1,r}$-subgraph-free graph $G$.  If $C$ is subcubic, then the problem can be solved in polynomial time. Otherwise, by Theorem~\ref{thm:s111r}, $C$ has bounded treedepth and again the problem can be solved in polynomial time. Finally, the solutions for its connected components can be merged in polynomial time.
\end{proof}


\begin{theorem} \label{thm:algo:s11qr}
Let $q$ and $r$ be positive integers.
A problem $\Pi$ can be solved in polynomial time on $S_{1,1,q,r}$-subgraph-free graphs if the following hold:
\begin{enumerate}[i)]
\item $\Pi$ can be solved in polynomial time on subcubic graphs,
\item $\Pi$ can be solved in polynomial time on graphs of bounded treedepth, and
\item $\Pi$ can be solved on graphs with proper bridges using a polynomial-time reduction to a family of instances on graphs that are either of bounded treedepth or subcubic.
\end{enumerate}
\end{theorem}
\begin{proof}
Let $H$ be one of the family of instances obtained from an instance $G$ of $\Pi$. 
As $H$ is either of bounded treedepth or subcubic, the problem can be solved in polynomial time.  As we have a reduction, once solved on all the family of instances, we can solve $\Pi$ on~$G$.
\end{proof}

The simplest way to apply Theorem~\ref{thm:algo:s11qr} is to show that if it is possible to solve~$\Pi$ on each of the family of components obtained by deleting the proper bridges of an instance, then these solutions combine to  provide a solution for the initial instance (since the components are quasi-bridgeless and so certainly either of bounded treedepth or subcubic by Theorem~\ref{thm:s11qr}.)


We now use Theorem~\ref{thm:algo:s11qr} to prove Theorem~\ref{cor:algo:s11qr}.  We do not apply Theorem~\ref{thm:algo:s111r} in this paper, as the results it would give us would just be special cases of those we have obtained using Theorem~\ref{thm:algo:s11qr}.  Nevertheless, there are potential applications of Theorem~\ref{thm:algo:s111r} as there might be C13 problems that can be solved componentwise but cannot be solved by finding the reduction required by Theorem~\ref{thm:algo:s11qr}.  We will see, in the proof below, that to solve {\sc Independent Feedback Vertex Set} via a reduction requires an intricate argument and the careful analysis of possible solutions on subcubic graphs that was provided by Theorem~\ref{t-subcubic2}.

\medskip
\noindent
{\bf Theorem~\ref{cor:algo:s11qr} (restated).}
{\it Let $q$ and $r$ be positive integers.
The following problems can be solved in polynomial time on $S_{1,1,q,r}$-subgraph-free graphs: {\sc Feedback Vertex Set}, {\sc Independent Feedback Vertex Set}, {\sc Connected Vertex Cover}, {\sc Colouring} and {\sc Matching Cut}.}

\begin{proof}[Proof of Theorem~\ref{cor:algo:s11qr}]
To show that the result follows immediately from Theorem~\ref{thm:algo:s11qr}, we can show that the problems can be solved by deleting bridges and considering the resulting graph componentwise; this will be trivial for some problems, but for others we will need to find a different reduction.  

For {\sc Feedback Vertex Set}, as bridges do not belong to cycles, the problem is unchanged when they are deleted.

For {\sc Independent Feedback Vertex Set} such a straightforward approach is not possible as if we simply delete bridges and solve the problem on the components, the merged solution might not be independent (since we might choose both endvertices of a deleted bridge).  We must argue a little more carefully.
Let $G$ be a $S_{1,1,q,r}$-subgraph-free graph and consider the treelike structure of $G$ when thinking of its blocks --- the connected components when the bridges are deleted.  In fact, consider a subgraph of $G$ that is a block plus all its incident bridges.  Some of these subgraphs might be subcubic; let us call these \emph{C-type}. For those that are not, we can assume, by Theorem~\ref{thm:s11qr}, that there is a constant $c$ such that their treewidth is at most $c$; let us call these subgraphs \emph{T-type} (note that this is a weaker claim that the Theorem~\ref{thm:s11qr} provides as we could assume that the treedepth was bounded).  If such a subgraph is both subcubic and has treewidth at most $c$, we will think of it as T-type.  We can assume $c \geq 3$ so a very nice cactus is T-type.  If subgraphs of the same type overlap (because they are joined by a bridge), we observe that their union is also of that type (since the union is also either, respectively, subcubic or of treewidth at most $c$).  So, merging overlapping subgraphs of the same type as much as possible we can consider $G$ as being made up of C- and T-type subgraphs and bridges that each join a C-type subgraph to a T-type subgraph.  As {\sc Independent Feedback Vertex Set} is a C13 problem we can solve it on these subgraphs.  Before we solve it on a C-type subgraph, we can delete pendant bridges (that link to a T-type subgraph in $G$) so the incident vertex now has degree at most~2.  As a very nice cactus is being considered as a T-type subgraph, we know, by Theorem~\ref{t-subcubic2}, that the solutions we find for C-type subgraphs do not use the vertices incident with the bridges.  Thus the solutions can be merged for a solution for $G$ that is also independent.

For {\sc Connected Vertex Cover}, let $G$ be a $S_{1,1,q,r}$-subgraph-free graph. Clearly, we may assume $G$ is connected, or it has no connected vertex cover.
As for {\sc Independent Feedback Vertex Set} consider each subgraph $J$ that is a block of $G$ and also include the bridges of $G$ incident with the block. Observe that $J$ is quasi-bridgeless and $S_{1,1,q,r}$-subgraph-free.  Noticing that a connected vertex cover $W$ of $G$ must contain both vertices incident with any proper bridge, we see that the restriction of  $W$ to the vertices of $J$ is a connected vertex cover of $J$ that includes vertices incident with bridges of $G$.
 And the construction of $J$ means its connected vertex covers will include these vertices adjacent to bridges in $G$.  Thus we see that have a reduction and can solve the problem on $G$.

For {\sc Colouring} if, for a graph $G$, we colour the components of the graph obtained by deleting bridges, then we can merge these into a colouring of $G$.  If the two endvertices of a bridge have been coloured alike, then we just permute the colours on one of the components.  This might create new clashes, but we move to the adjacent components and permute there.  By the definition of bridge, we will never have to permute the colours on a component more than once so the process terminates.

For {\sc Matching Cut}, if a graph contains a bridge, then we have immediately that it is a yes instance.
\end{proof}


\section{Graphs Excluding Subdivided Stars as a Subgraph: Hardness}
\label{s-hardness}



We prove Theorem~\ref{thm:hardness}.

\medskip
\noindent
{\bf Theorem~\ref{thm:hardness} (restated).}
{\it {\sc Feedback Vertex Set} and {\sc Independent Feedback Vertex Set} are \NP-complete on the class of  $S_{2,2,2,2}$-subgraph free graphs that have maximum degree~$4$.}

\begin{proof}
	\begin{figure}[htbp]
		\centering
		\includegraphics[width=0.75\textwidth]{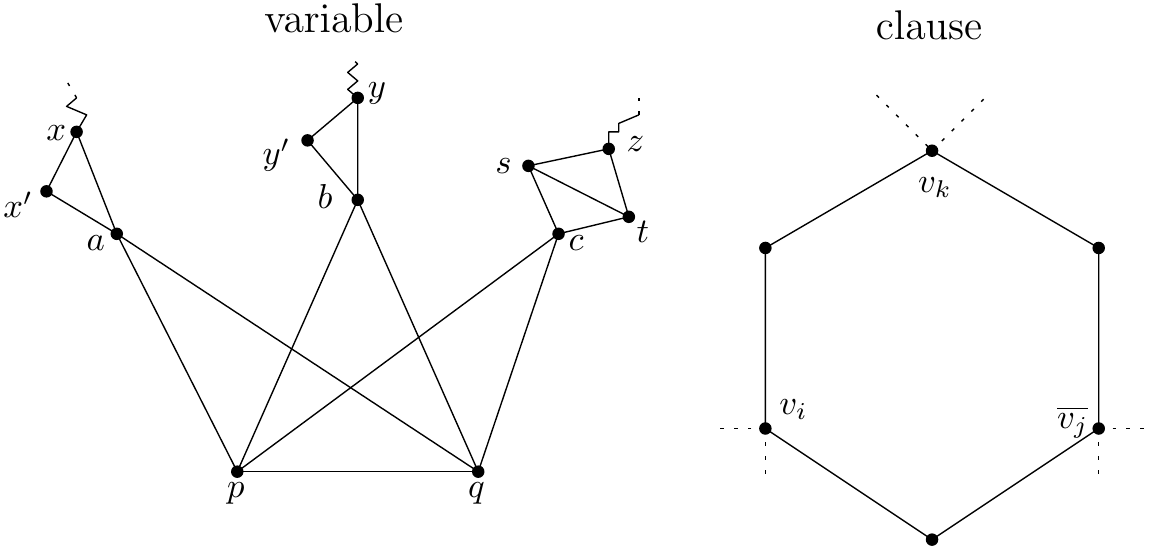}
		\caption{The variable and clause gadgets (for clauses of size 3) from the proof of Theorem~\ref{thm:hardness}. The vertices $x$, $y$, and $z$ of a variable gadget will be identified with the (labeled) vertices of clause gadgets.} \label{fig:Construction}
	\end{figure}

	Both problems belong to \NP. We shall show a reduction from the following \NP-complete problem {\sc 2P1N-3SAT}~\cite{DJPSY94}. 
	
	\problemdef{{\sc 2P1N-3SAT}}{A CNF formula $\Phi$ where each clause contains at most three literals and each variable occurs twice positively and once negatively.}{Does $\Phi$ have a satisfying assignment?}
	
	\medskip
	\noindent Given an instance of {\sc 2P1N-3SAT} with on variables $\{v_1, \ldots, v_n\}$, we construct a graph $G$ as follows. For each variable $v_i$, we construct the gadget shown in Figure~\ref{fig:Construction}. The triangles $xx'a$ and $yy'b$ represent the positive occurrences of the variable, while the diamond $zstc$ represents the negative occurrence. For each clause $C_j$, we construct a hexagon if the clause has size~$3$ and a square if the clause has size~$2$ (we may assume that no clause has size~$1$). Alternate vertices of this clause gadget represent literals and are identified with a vertex $x$, $y$ or $z$ of the corresponding variable gadget.  Clearly this can be done in such a way each vertex $x$ and $y$ of each variable gadget is identified with exactly one vertex from a clause gadget that represents a positive literal and each vertex $z$ of each variable gadget is identified with exactly one vertex from a clause gadget that represents a negative literal.
		Note that $G$ has maximum degree~$4$.

	\begin{claim}\label{clm:S2222-free}
		$G$ does not contain $S_{2,2,2,2}$ as a subgraph.
	\end{claim}
	
		\begin{claimproof}
Let us consider where we might find the centre vertex of a $S_{2,2,2,2}$ in $G$.  
Clearly a vertex $v$ cannot be the centre vertex if its 2-neighbourhood in $G$ contains a cut of size~3 (that is, if there are 3 vertices each of distance at most 2 from $v$ that form a cut in $G$).
 The centre vertex cannot be the vertices $p$ or $q$ of a variable gadget, because the set $\{a,b,c\}$ of the same gadget forms a cut of size~$3$ in the $2$-neighbourhood of $p$ and $q$. The centre vertex cannot be the vertices $a$, $b$, or $c$ of a variable gadget either, because $\{x,p,q\}$, $\{y,p,q\}$ and $\{z,p,q\}$ respectively form cuts of size~$3$ in their $2$-neighbourhoods. The vertices $x$, $y$, and $z$ cannot be the centre vertex as in their $2$-neighbourhood is a cut of size 3 that contains their two neighbours in a clause gadget and, respectively, $a$, $b$ and $c$.   The remaining vertices of $G$ have degree less than 4.  The claim is proved.
	\end{claimproof}

	Any feedback vertex set of a variable gadget has size at least~$4$, because it contains four disjoint cycles.
	So any feedback vertex set of $G$ must contain at least $4n$ vertices.  It only remains to show that
		$G$ has an (independent) feedback vertex set of size at most $4n$ if and only if $\Phi$ is satisfiable.

	Assume that $\Phi$ has a satisfying assignment. We construct a feedback vertex set $F$ of $G$. If a variable is true, then the vertices $x$, $y$, $p$, and $t$ of the  variable gadget belong to $F$. If a variable is false, then instead $z$, $a$, $b$, and $c$ belong to $F$.  Thus $F$ is an independent set  (vertices of distinct variable gadgets are not adjacent) and its size is exactly~$4n$.
	
	\begin{claim} \label{clm:fvs}
	 $F$ is a feedback vertex set.  
\end{claim}

\begin{claimproof}
Notice that if a literal of a clause is satisfied, then, in the clause gadget, the corresponding vertex is in $F$.  Thus, as clause is satisfied, each cycle contained in a single variable or clause gadget contains a vertex of $F$.   Consider a cycle of $G$ that is not contained within a single gadget.  It must include a non-trivial path of some variable gadget where the endvertices are two of $\{x,y,z\}$.  If it includes $x$ it must also include $a$ and if it includes $y$ it must also include $b$.  But $F$ contains one of $\{x,a\}$ and one of $\{y,b\}$ so such a cycle also intersects $F$. Thus $F$ intersects all the cycles of $G$. 
\end{claimproof}

	Conversely, suppose that $G$ has a feedback vertex set $F$ of size at most $4n$. Again, each variable gadget contains at least four vertices of $F$ and so  contains exactly four vertices of $F$.  Notice that $F$ cannot contain either $\{x,z\}$ or $\{y,z\}$ as, in each case, there are three disjoint cycles of the gadget that would need to be covered by just two vertices.  
	
Let us describe a satisfying assignment of $\Phi$.  If, for a variable gadget, either $x$ or $y$ belongs to $F$, we let the variable be true. If $z$ belongs to $F$, we let it be false.  By the preceding argument, there is no possibility that we must set a variable to be both true and false.   If none of $\{x,y,z\}$ belong to $F$, we set the value of the variable arbitrarily.  This is a satisfying assignment as every clause gadget (which is a cycle) must have at least one vertex in $F$ and the corresponding variable is satisfied. 
\end{proof}

\section{Proofs of the Classifications}  \label{s-proofs}

We prove Theorems~\ref{t-dichofvs}--\ref{t-dichomc}. 
Noting that the theorems contain some analogous results, and wishing to avoid repetition, we make a few general comments that apply to all proofs.

We state again that the five problems under consideration are C13 problems.  Thus when $H \in {\cal S}$, each theorem follows from Theorem~\ref{t-dicho3}.  When $H = S_{1,1,q,r}$, we apply Theorem~\ref{cor:algo:s11qr}.  Thus, except for Theorem~\ref{t-dichocol} on {\sc Colouring}, the following proofs need only cover the \NP-complete cases.

\begin{proof}[Proof  of Theorem~\ref{t-dichofvs}]
We note again that {\sc Feedback Vertex Set} reduces to {\sc Independent Feedback Vertex Set} after subdividing each edge so here we consider only the former.

By Poljak’s construction~\cite{Poljak1974}, for every integer $g \geq 3$, {\sc Feedback Vertex Set} is \NP-complete for graphs of girth at least $g$ (the girth of a graph is the length of its shortest cycle).  Thus {\sc Feedback Vertex Set} is \NP-complete for $H$-subgraph-free graphs whenever $H$ contains a cycle.

Suppose that $H$ has $m$ vertices and more than one vertex of degree at least $3$.  From any graph $G$, if we we subdivide each edge $m$ times, we obtain a graph $J$ that is $H$-subgraph free since the distance between any pair of vertices of degree more than $2$ is at least $m+1$.  In finding in a minimum size feedback vertex set of $J$, we may as well restrict ourselves to selecting vertices of $G$.  This implies that {\sc Feedback Vertex Set} is \NP-complete for $H$-subgraph-free graphs.
 
The problem is \NP-complete on planar graphs of maximum degree~$4$~\cite{Sp83} (so for $K_{1,5}$-subgraph-free graphs). 
 
Theorem~\ref{thm:hardness} completes the proof.
\end{proof}


\begin{proof}[Proof  of Theorem~\ref{t-dichocvc}]
For every integer $g \geq 3$, {\sc Connected Vertex Cover} is \NP-complete for graphs of girth at least $g$~\cite{Munaro17}, so also for $H$-subgraph-free graphs whenever $H$ contains a cycle.  It is \NP-complete on graphs of maximum degree~4~\cite{GJ79}, so for $K_{1,5}$-subgraph-free graphs. 
\end{proof}

\begin{proof}[Proof  of Theorem~\ref{t-dichocol}]
For every integer $g \geq 3$, {\sc Colouring} is \NP-complete for graphs of girth at least $g$~\cite{MP96}, so also for $H$-subgraph-free graphs whenever $H$ contains a cycle.
In~\cite{GJS76}, it was shown that {\sc Colouring}  is \NP-complete on (planar) graphs of maximum degree~$4$, and so too for $K_{1,5}$-subgraph-free graphs.
The other cases are all proved in~\cite{GPR15}
\end{proof}

\begin{proof}[Proof  of Theorem~\ref{t-dichomc}]
For every integer $g \geq 3$, {\sc Matching Cut} is \NP-complete for graphs of girth at least $g$~\cite{FLPR23}, so also for $H$-subgraph-free graphs whenever $H$ contains a cycle.  It is \NP-complete on graphs of maximum degree~$4$~\cite{Ch84}, so for $K_{1,5}$-subgraph-free graphs. 
\end{proof}

\section{Conclusions}\label{s-con}

We made significant progress towards classifying the complexity of five well-known C13-problems on $H$-subgraph-free graphs, extending previously known results. In particular, we identified a gap in the literature, and provided a polynomial-time algorithm for {\sc Independent Feedback Vertex Set} for subcubic graphs.

If $H$ is connected, then we narrowed the gap for these problems to the open case where $H=S_{1,p,q,r}$, so $H$ is a subdivided star with one short leg and three arbitrarily long legs.
To obtain a result for connected $S_{1,p,q,r}$-subgraph-free graphs similar to our previous results, we would need the graphs to be $3$-edge-connected. Indeed, the statement is false without this assumption. Consider the class of all graphs that are each the union of  a path and a $K_{1,4}$ two of whose leaves are identified with distinct end-vertices of the path and whose other two leaves are made adjacent.  This is a class of graphs that are  bridgeless, not subcubic and $S_{1,p,q,r}$-subgraph-free and again has unbounded treedepth.  It is not yet clear whether a suitably modified theorem statement would indeed hold. In addition, it is unclear whether this would yield a result that could be applied in the same way as Theorems~\ref{thm:s111r} and~\ref{thm:s11qr} were above. We leave the case $H=S_{1,p,q,r}$ as future research.

Finally, we also leave determining the complexity of {\sc Connected Vertex Cover} and {\sc Matching Cut} on $S_{2,2,2,2}$-subgraph-free graphs as an open problem.



\bibliography{mybib3}




\end{document}